\documentclass[smallextended]{svjour3}
\smartqed
\usepackage[margin=1in]{geometry}
\usepackage[onehalfspacing]{setspace}
\usepackage{amsmath,mathtools}

\usepackage{algorithm}
\usepackage{algorithmic}

\usepackage[english]{babel}

\usepackage{pdflscape}
\usepackage{fancyvrb}
\usepackage{calc}
\usepackage{rotating}
\usepackage{pgf,tikz}
\usepackage{mathrsfs}
\usetikzlibrary{arrows}
\usepackage{mathtools}
\DeclarePairedDelimiter\ceil{\lceil}{\rceil}

\usepackage{listings}
\usepackage[]{qcircuit}
\usepackage{braket}
\usepackage{mathtools}
\usepackage{varwidth}
\usepackage{caption}
\usepackage{subcaption}
\usepackage{graphicx}
\usepackage{hyperref,doi}

\makeatletter
\def\paragraph{\@startsection{paragraph}{4}%
  \z@\z@{-\fontdimen2\font}%
  {\normalfont\bfseries}}
\makeatother

\usepackage{graphicx}
\usepackage{pgffor}
\usepackage{tikz,tikz-cd}
\usetikzlibrary{backgrounds,fit,decorations.pathreplacing}  
\usepackage{booktabs}
\usepackage{enumerate}

\usepackage{amssymb, amsmath, amsthm}

\usepackage{xypic}
\usepackage{scalerel}
\usepackage{ulem}

\usepackage{graphicx}              
\usepackage{amsmath}               
\usepackage{amsfonts}              
\usepackage{amsthm}                
\usepackage{xkeyval}               
\usepackage{amssymb}
\usepackage{enumerate}             
\usepackage{color}
\usepackage{breqn}                 
\usepackage{bm}                    

\allowdisplaybreaks[1]
\usepackage{nicefrac}
\usepackage{booktabs}

\usepackage{kpfonts}
\usepackage{stackengine}
\usepackage{calc}
\newlength\shlength
\newcommand\xshlongvec[2][0]{\setlength\shlength{#1pt}%
	\stackengine{-5.6pt}{$#2$}{\smash{$\kern\shlength%
			\stackengine{7.55pt}{$\mathchar"017E$}%
			{\rule{\widthof{$#2$}}{.57pt}\kern.4pt}{O}{r}{F}{F}{L}\kern-\shlength$}}%
	{O}{c}{F}{T}{S}}

\usepackage{hyperref}
\hypersetup{
	colorlinks   = true,    
	citecolor    = red      
}

\newcommand{\ds}[1]{\displaystyle{#1}}  

\newcommand{\RN}[1]{%
  \textup{\uppercase\expandafter{\romannumeral#1}}%
}

\allowdisplaybreaks

\newcommand{\meqref}[1]{Eq.~$\! \eqref{#1} $}
\newcommand{\mref}[1]{Sect.~$ \!\ref{#1} $}
\newcommand{\mfig}[1]{Fig.~$ \!\ref{#1} $}

\newtheorem{thm}{Theorem}[section]

\newtheorem{lem}[thm]{Lemma}

\newtheorem{defn}[thm]{Definition}

\newcommand{\RR}{\mathbb{R}}      

\def\OO{{\mathbb O}}

\def\RR{{\mathbb R}}

\def\ll{{\frak l}}  

\def\ss{{\frak s}}
\def\tt{{\frak t}}

\def\<{\langle}
\def\>{\rangle}

\def\sl{\ss\ll}

\DeclareMathOperator{\Cost}{Cost}
\DeclareMathOperator{\SolutionPaths}{SolutionPaths}

\numberwithin{equation}{section}

\begin{document}

	\title{A novel quantum grid search algorithm and its application 
	}
	
	
	\author{Alok Shukla       \and
		Prakash Vedula 
	}
	
	
	\institute{Alok Shukla  \at
		The University of Manitoba, Canada \\
		\email{sajal.eee@gmail.com}           
		\and
		Prakash Vedula \at
	   The University of Oklahoma, Norman, USA \\
		\email{pvedula@ou.edu}           
}

	\date{}

	\maketitle
	
	\begin{abstract}
In this paper we present a novel quantum algorithm, namely the quantum grid search algorithm, to solve a special search problem.  Suppose $ k $ non-empty buckets are given, such that each bucket contains some marked and some unmarked items. In one trial an item is selected from each of the $ k $ buckets. If every selected item is a marked item, then the search is considered successful. This search problem can also be formulated as the problem of finding a ``marked path'' associated with specified bounds on a discrete grid. Our algorithm essentially uses several Grover search operators in parallel to efficiently solve such problems. We also present an extension of our algorithm combined with a binary search algorithm in order to efficiently solve global trajectory optimization problems. Estimates of the expected run times of the algorithms are also presented. We note that computational gain in our algorithm comes at the cost of increased complexity of the quantum circuitry.

	\end{abstract}

\section{Introduction}\label{sec:intro}

A powerful quantum search algorithm was presented by Grover \cite{grover1996fast} to find a marked item in a given $ N $-item database, in $ \OO(\sqrt{N}) $ evaluations of the oracle function $ f $. Here the oracle function $ f $ is a black-box function such that $ f(x)=1 $ if  $ x $ is a marked item, otherwise $ f(x)=0 $. Grover's algorithm offers a quadratic speedup, as clearly a classical search in this situation requires $ \OO(N) $ function calls. Although, a quadratic speedup is impressive, it is natural to ask if it can further be improved. However, it is well-known that Grover's quantum search algorithm is optimal \cite{articleZalka}.
In this paper, we present an algorithm and describe its applications in solving a class of optimization problems, where exponential gains could be achieved in comparison to non-parallel classical algorithms. We remark that if we consider parallel classical algorithms then the computational speedup is only quadratic. 

Now we describe a formulation of the search problem that we solve in this work. Consider $ k $ non-empty buckets, labeled $ X_1,X_2,\cdots X_k $, with each bucket containing some marked and some unmarked items. Suppose one item is taken out blindfolded from each bucket, say $ x_i $ is taken out from $ X_i $ for $ i=1,2,\cdots k $. If each $ x_i $ is marked, then the search is considered successful. 
To illustrate a practical example, we consider the blending problem in which a product is made by blending various components in different quantities. More explicitly, suppose $ X_1,X_2,\cdots X_k $ are the $ k $ components needed to manufacture the product. Further, suppose that in the manufacturing process, for $ i=1,2,\cdots k $, the quantity of the component $ X_i $ used is chosen from the set $ X_{i,1},X_{i,2}, \cdots X_{i,n_i} $. Here $ X_i $ is a bucket and $ X_{i,1},X_{i,2}, \cdots X_{i,n_i} $ are the items in the bucket $ X_i $.  We assume that there exists an oracle function $ f $ which assigns a number denoting the quality of the finished product, for any given choice of the quantities of the components selected. The goal of the optimization problem is to select the quantities of the components to manufacture a product of desired quality, i.e., to find $ x_1,x_2,\cdots x_k $, with $ x_i $ being the quantity of the component $ X_i $ chosen such that $ m \leq f(x_1,x_2,\cdots x_k) $. Here $ m $ denotes the minimum threshold of the quality for the product to be accepted. Suppose only one item is marked in each bucket. Then, clearly in the worst case a non-parallel classical algorithm will need $ \prod_{i=1}^{k} n_i $ searches to discover the desired composition of the product. However, our proposed algorithm will find success within the expected run time of the order of $ \OO\left(\max\left(\sqrt{ n_i} \right)_{i=1}^{i=k} \right) $. 

We present another important class of problems which can be efficiently solved by our proposed algorithms. This is the global trajectory optimization problem using the discretization framework presented in \cite{shuklavedula2019}. We note that in \cite{shuklavedula2019} we used quantum algorithms based on Grover's search algorithm for solution of global trajectory optimization problems, and demonstrated that a quadratic speedup could be achieved in comparison to classical search based algorithms. Our proposed algorithms (Algorithm~$\ref{alg:grid} $ and Algorithm~$ \ref{alg:binary} $) in the present work offer further improvements over the quantum search algorithms presented in \cite{shuklavedula2019}. Further, we recall that in \cite{shuklavedula2019}, we presented a discretization scheme, wherein both the independent variables (such as time)  and the dependent variables were discretized. Our discretization scheme involved the selection of a finite number of coarse-grained states, with a cost function defined on each state. Thereby, the continuous  optimization problem was transformed to a discrete search problem on the discretization grid, in which a state with the minimum associated cost is to be determined. Further, the discretization scheme in \cite{shuklavedula2019} treated discretization in both the physical space as well as the coefficient space, and it also provided mathematical bounds using the coefficients of Chebyshev polynomials while discretizing in the coefficient space. We refer readers to \cite{shuklavedula2019} for further details on the discretization scheme. 
Once discretization is carried out using the framework presented in \cite{shuklavedula2019}, the resulting discrete search problem is amenable to treatment by our proposed quantum grid search algorithms. One additional assumption is the availability of local oracles, which will be described later. 

As noted earlier, the key idea in our proposed Algorithm~$\ref{alg:grid} $ is the use of several Grover search operators in parallel to solve the discrete search problem efficiently. Algorithm~$\ref{alg:grid} $ gives a solution within the specified bounds. Algorithm~$ \ref{alg:binary} $  uses Algorithm~$\ref{alg:grid} $ along with the classical binary search on the cost function, to efficiently solve global trajectory optimization problems, as accurately as required. 

In the first half of the paper, we describe the quantum grid search algorithm, and in the second half of this work, we give its application in solving trajectory optimization problem. More specifically, in Sect.~$ \ref{sect:Quantum} $ we briefly recall Grover's search algorithm, and subsequently in Sect.~$ \ref{sect:grid_search} $ we describe our grid search algorithm. Further, to describe the application of quantum grid search algorithm, we present a formulation of  the general trajectory optimization problem, in Sect.~$ \ref{sect:problem formulation} $. Then in Sect.~$ \ref{sect:application} $, we describe how quantum grid search algorithm, Algorithm~$\ref{alg:grid} $, can be combined with Algorithm~$ \ref{alg:binary} $ to solve a given trajectory optimization problem with greater efficiency. Finally, we present our concluding remarks in \mref{sect:Conclusion}.  

\section{\textbf{Quantum search: Grover's algorithm }} \label{sect:Quantum}

Quantum computing is an exciting field which nicely combines quantum physics with computer science. There has been a lot of progress not only in theoretical development of quantum algorithms, but also in their practical implementation on quantum computers. We refer the readers to any standard book on the subject for a more in-depth treatment (for eg., \cite{nielsen2002quantum}, \cite{rieffel2011quantum} or \cite{yanofsky2008quantum}).  

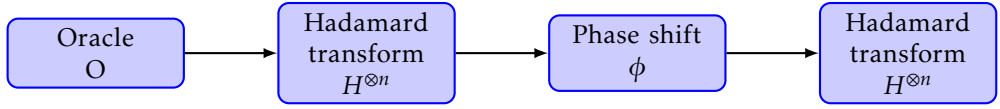
\begin{figure}
	\centering
	\begin{tikzpicture}[%
	scale= 0.8,
	block/.style={
		rectangle,
		draw=blue,
		thick,
		fill=blue!20,
		text width=6em,
		align=center,
		rounded corners,
		minimum height=2em
	},
	block1/.style={
		rectangle,
		draw=blue,
		thick,
		fill=blue!20,
		text width=5em,
		align=center,
		rounded corners,
		minimum height=2em
	},
	line/.style={
		draw,thick,
		-latex',
		shorten >=2pt
	},
	cloud/.style={
		draw=red,
		thick,
		ellipse,
		fill=red!20,
		minimum height=1em
	}
	]
	\path (-9,0) node[block] (G) {Oracle \\ O}
	(-4.5,0) node[block] (H) {Hadamard transform \\ $ H^{\otimes n} $ }   
	(0,0) node[block] (I) {Phase shift \\ $\phi$}  
	(4.5,0) node[block] (J) {Hadamard transform \\ $ H^{\otimes n} $};    
	\draw[-latex,thick] (G) -- (H);                                   
	\draw[-latex,thick] (H) -- (I);  
	\draw[-latex,thick]   (I)--(J); 
	\end{tikzpicture}
	\caption{The Grover operator $ G $} \label{fig_Grover's_algo}
\end{figure}

Here we will briefly recall Grover's quantum search algorithm \cite{grover1996fast} for finding one or more marked items out of a list of $ N $ given items. 
Access to an oracle $ O $, which is a black-box function, is assumed. This black-box function is  such that for any item with index $ x $ in the list, i.e., with $ 0 \leq x \leq N-1 $, we have  $ f(x)=1 $ if $ x $ is a marked item, otherwise $ f(x)=0 $.

For simplicity, we assume that $ N = 2^n $ and so the discrete search space can be represented by a $ n $ qubit system. The algorithm begins with the system in $ \ket{0}^{\otimes n} $. 
Then the Hadamard transform $ H^{\otimes n} $ is applied on the input  $ \ket{0}^{\otimes n} $  to transform it to a uniform superimposed state, say $ \ket{\psi} $, such that
\begin{align} \label{Eq:Hadamard}
\ket{\psi} = H^{\otimes n} (\ket{0}^{\otimes n} ) = \frac{1}{\sqrt{N}} \sum_{x=0}^{N-1} \ket{x}.
\end{align}
We note that at the heart of Grover's algorithm lies the Grover operator $ G $ (\mfig{fig_Grover's_algo}), which is essentially the following unitary transformation 
\begin{align} \label{Eq_Grover_operator}
G = \left( H^{\otimes n} (2 \ket{0} \bra{0} - I) H^{\otimes n} \right) O = \left(2 \ket{\psi} \bra{\psi} - I  \right) O.
\end{align}
The phase shift $ \phi $, in \mfig{fig_Grover's_algo}, is the unitary operator $ (2 \ket{0} \bra{0} - I )$, which transforms all non-zero basis state $ \ket{x} $ to $ -\ket{x} $ with $ \ket{0} $ remaining unchanged.
 In fact, as noted earlier  $ G = \left(2 \ket{\psi} \bra{\psi} - I  \right) O$. The action of oracle $ O $ amounts to changing the phase of the marked item. On the other hand, it can be checked that 
\begin{align} \left(2 \ket{\psi} \bra{\psi} - I  \right) (\sum_{k=0}^{N-1} \alpha_k  \ket{k})  = \sum_{k=0}^{N-1} \left( -\alpha_k + 2 m \right) \ket{k} \end{align}
where $ m = \frac{1}{N} \sum_{k=0}^{N-1} \alpha_k $ is the mean of $ \alpha_k $.  Hence, the action of $ \left(2 \ket{\psi} \bra{\psi} - I  \right) $ is essentially an inversion about the mean. Therefore, one iteration of Grover operator results in amplification of the amplitude of the marked item (See \mfig{Fig_Grover_operator_amplitude}). 

\definecolor{qqqqff}{rgb}{0.,0.,1.}
\definecolor{xdxdff}{rgb}{0.49019607843137253,0.49019607843137253,1.}
\definecolor{uuuuuu}{rgb}{0.26666666666666666,0.26666666666666666,0.26666666666666666}
\begin{figure}
	\centering
	\begin{tikzpicture}[line cap=round,line join=round,>=triangle 45,x=0.5535055350553507cm,y=0.9385665529010234cm]
	\def\shift{3} 
	\def\down{-4} 
	\def \x {7}
	\draw [-] (0.,0.) -- (0.,2.);
	\draw [->] (0.,0.) -- (0.,1.);
	\draw [->] (1.,0.) -- (1.,1.);
	\draw [->] (2.,0.) -- (2.,1.);
	\draw [->] (3.,0.) -- (3.,1.);
	\draw [->] (4.,0.) -- (4.,1.);
	\draw [->] (5.,0.) -- (5.,1.);
	\draw [->] (6.,0.) -- (6.,1.);
	\draw [->] (7.,0.) -- (7.,1.);
	\draw (0.,0.)-- (9.,0.);
	\draw (-1,1.2) node[anchor=north] {$ \frac{1}{\sqrt{N}} $};
	\draw (0,-0.12) node[anchor=north] {0};
	\draw (1,-0.12) node[anchor=north] {1};
	\draw (6.75,-0.12) node[anchor=north west] {$N-1$};
	\draw (3.76,-0.12) node[anchor=north west] {$k$};
	\begin{scriptsize}
	\end{scriptsize}
	\draw [-] (\shift+ 10.,0.) -- (\shift+10.,2.);
	\draw [->] (\shift+10.,0.) -- (\shift+10.,1.);
	\draw [->] (\shift+11.,0.) -- (\shift+11.,1.);
	\draw [->] (\shift+12.,0.) -- (\shift+12.,1.);
	\draw [->] (\shift+13.,0.) -- (\shift+13.,1.);
	\draw [->] (\shift+14.,0.) -- (\shift+14.,-1.);
	\draw [->] (\shift+15.,0.) -- (1\shift+5.,1.);
	\draw [->] (\shift+16.,0.) -- (\shift+16.,1.);
	\draw [->] (\shift+17.,0.) -- (\shift+17.,1.);
	\draw (\shift+10.,0.)-- (\shift+19.,0.);
	\draw (\shift+10,-0.12) node[anchor=north] {0};
	\draw (\shift+11,-0.12) node[anchor=north] {1};
	\draw (\shift+16.75,-0.12) node[anchor=north west] {$N-1$};
	\draw (\shift+17.0,1.3) node[anchor=north west] {mean};
	\draw [dotted] (\shift+10.,0.9) -- (\shift+19.,0.9);
	\draw (\shift+13.9,-0.12) node[anchor=north west] {$k$};
	
	\draw [-] (\x + 0.,\down + 0.) -- (\x + 0.,\down + 2.);
	\draw [->] (\x + 0.,\down + 0.) -- (\x + 0.,\down + 1.);
	\draw [->] (\x + 1.,\down + 0.) -- (\x + 1.,\down + 1.);
	\draw [->] (\x + 2.,\down + 0.) -- (\x + 2.,\down + 1.);
	\draw [->] (\x + 3.,\down + 0.) -- (\x + 3.,\down + 1.);
	\draw [->] (\x + 4.,\down + 0.) -- (\x + 4.,\down + 3.0);
	\draw [->] (\x + 5.,\down + 0.) -- (\x + 5.,\down + 1.);
	\draw [->] (\x + 6.,\down + 0.) -- (\x + 6.,\down + 1.);
	\draw [->] (\x + 7.,\down + 0.) -- (\x + 7.,\down + 1.);
	\draw (\x + 0.,\down + 0.)-- (\x + 9.,\down + 0.);
	\draw (\x + -1,\down + 1.2) node[anchor=north] {$ \frac{1}{\sqrt{N}} $};
	\draw (\x + 0,\down -0.12) node[anchor=north] {0};
	\draw (\x + 1,\down -0.12) node[anchor=north] {1};
	\draw (\x + 6.75,\down -0.12) node[anchor=north west] {$N-1$};
	\draw (\x + 3.76,\down -0.12) node[anchor=north west] {$k$};
	\draw (\x + 4.1,\down + 3) node[anchor=north west] {$\frac{3}{\sqrt{N}}$};
	\end{tikzpicture}
	\caption {Amplitude amplification of the marked item resulting from one iteration  of Grover operator. } \label{Fig_Grover_operator_amplitude}
\end{figure}
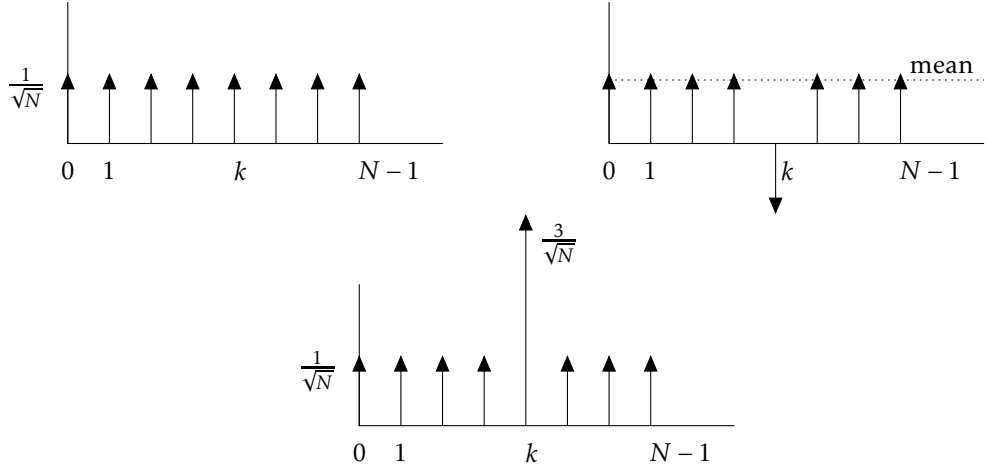

We note that if $ M=1 $, then the Grover operator $ G $ is applied $ R \approx \frac{\pi}{4} \sqrt{N} $ times. After $ R $ iteration of $ G $, the measurement of the system yields the marked item. To summarize, Grover's algorithm for the case $ M=1 $ is as follows.

\begin{algorithmic}[1]

	\STATE \textbf{Initialization:} Initialize the input $ n $ qubits as $ \ket{\psi} = \ket{0}^{\otimes n} $.
	
	\STATE \textbf{Iteration:} Apply $ H^{\otimes n} $ on $\ket{\psi} $.
	
	\STATE \textbf{Grover operator:} Apply Grover operator $ G $ a total of $R \approx \frac{\pi}{4} \sqrt{N}$ times.
	
	\STATE \textbf{Measurement:} Measure the $ n $ qubits and obtain candidate solution $x$.
	
\end{algorithmic}

We refer readers to \cite{boyer1998tight}, \cite{articleZalka} for a tight analysis of this algorithm.

\subsection{\textbf{Unknown number of marked items}}
We consider the case when the number of marked items is unknown. Suppose $ M $ items out of $ N $ are marked, with $ M $ is not known to the algorithm in advance. Let $ \alpha =  \frac{1}{\sqrt{M}} $ and $ \beta = \frac{1}{\sqrt{N-M}} $.
The initial state $ \ket{\psi}  $  is given by $$ \ket{\psi}  =  \sum_{x,\, x \text{ marked}} \, \alpha \ket{x} + \sum_{x,\, x \text{ not marked}} \, \beta \ket{x}  .$$
It can be verified that after  $ j $ iteration of the Grover operator $ G $, the initial state $ \ket{\psi} $  is transformed to 
\begin{align}
G^j\, \ket{\psi}  =  \sum_{x,\, x \text{ marked}} \, \alpha \sin\left((2j+1\right)\theta) \,\ket{x} + \sum_{x,\, x \text{ not marked}} \, \beta \cos\left((2j+1\right)\theta) \, \ket{x},
\end{align}\label{eq_Grover_iteration}
where
$\theta  = \sin^{-1}(\sqrt{\frac{M}{N}})$.

\section{Quantum grid search algorithm}\label{sect:grid_search}

We recall the search problem that we proposed to solve in \mref{sec:intro}. We consider $ k $ non-empty buckets labeled $ X_1,X_2,\cdots X_k $, with each bucket containing some marked and some unmarked items. Suppose one item is taken out blindfolded from each bucket at the same time. Then the search is considered successful is all items picked  are marked. Here, we consider an alternate formulation of the search problem, with the intention of applying this algorithm to solve trajectory optimization problems.
In this alternate formulation, our quantum grid search algorithm solves the problem of finding ``a marked path'' on a grid out of many such possible ``paths''. More explicitly, 
for $ i=1 $ to $ i=k$,  we define $ X_i $ as follows. Let $ X_i = \{x_{ij} \}$,  with $ j $ running from $ 1 $ to $ n_i $,  be a set of cardinality $ n_i $. We also assume that out of $ n_i $ elements, a fixed but unknown number of $ m_i $ elements are marked. A ``path'' is an element of the set $ X= \prod_{i=1}^{k} X_i $ and a path $ (y_1,y_2,\cdots y_k) $ will be called ``a marked path'' if $ y_i $ is a marked element in $ X_i $ for $ i=1,2,\cdots k $.

We  assume that, for $ i=1 $ to $ i=k $, local oracle function $ f_i $ is available such that $ f_i(t) =1 $, if $ t  $ is a marked item in $ X_i $, i.e., 
\begin{align}\label{Eq:localoracle}
f_i(t) = \begin{cases}
& 1  \qquad \text{if the element $t$ is a marked element in $X_i$, }\\
&  0 \qquad \text{otherwise.}
\end{cases}
\end{align}
The objective of the quantum grid search algorithm is to find a marked ``path'', i.e., to find a tuple \\ $ (y_1,y_2, \cdots y_k) $ with $ y_i \in X_i $ for $ i=1 $ to $ k $, such that 
$ f(y_1,y_2,\cdots y_k) = 1 $.

Using the local oracle functions defined above, one can obtain a global oracle function $ f  $ from $ X $ to $ \{0,1\} $ such that 
$ f(y_1,y_2,\cdots y_k) = 1 $ if and only if  $ y_i $ is a marked element in $X_i$ for $ i=1,2,\cdots k $; otherwise $ f(y_1,y_2,\cdots y_k) =0 $, i.e.,
\begin{align*}
f(y_1,y_2,\cdots y_k) = \begin{cases}
& 1  \qquad \text{if, for $ i=1$ to $ k $, the element $ y_i$ is a marked element in $X_i$, }\\
&  0 \qquad \text{otherwise.}
\end{cases}
\end{align*}

A schematic diagram of the quantum circuit for \textit{Algorithm \ref{alg:grid}}, for $ k=3 $, is given in \mfig{fig:circuit}. We note that essentially the main components of the circuit are $ k=3 $ Grover operators acting in parallel. 

\begin{figure} 
	\centerline{
		\begin{tikzpicture}[thick]
		%
		\tikzstyle{operator} = [draw,fill=white,minimum size=1.5em] 
		\tikzstyle{phase} = [fill,shape=circle,minimum size=3.5pt,inner sep=0pt]
		\tikzstyle{surround} = [fill=blue!10,thick,draw=black,rounded corners=2mm]
		%
		\node at (-1.5,0) (q11) {$ X_1 $};
		\node at (-1.5,-1) (q22) {$ X_2 $};
		\node at (-1.5,-2) (q33) {$ X_3 $};   
		\node at (-1.25,0) (q1) {$  $};
		\node at (-1.25,-1) (q2) {$  $};
		\node at (-1.25,-2) (q3) {$  $};    
		\node[operator] (op11) at (1,0) {G} edge [-] (q1);
		\node[operator] (op21) at (1,-1) {G} edge [-] (q2);
		\node[operator] (op31) at (1,-2) {G} edge [-] (q3);
		\node[phase] (phase11) at (3,0) {} edge [-] (op11);
		\node[phase] (phase12) at (3,-1) {} edge [-] (op21);
		\node[phase] (phase13) at (3,-2) {} edge [-] (op31);
		\node[phase] (phase110) at (-0.25,0) {} edge [-] (op11);
		\node[phase] (phase120) at (-0.25,-1) {} edge [-] (op21);
		\node[phase] (phase130) at (-0.25,-2) {} edge [-] (op31);
		\node[phase] (phase110) at (2.25,0) {} edge [-] (op11);
		\node[phase] (phase120) at (2.25,-1) {} edge [-] (op21);
		\node[phase] (phase130) at (2.25,-2) {} edge [-] (op31);
		\draw [thick] (3,0.25) -- (5,0.25) -- (5,-2.25) -- (3,-2.25)--(3,0.25); 
		\node at (4,-1) {$ f $};
		\draw [thick] (5,-1)--(6,-1);
		\node[phase,thick] (phase22) at (5,-1) {} edge [-] (5.5,-1);
		\node at (0,0) (q100) {$ $};
		\node at (2,0) (q200) {$ $};
		\begin{pgfonlayer}{background} 
		\node[surround] (background) [fit = (q100) (op11) (op21) (op31) (q200) ] {};
		\end{pgfonlayer}
		\end{tikzpicture}
	}
	\caption{
		Schematic of the quantum circuit for \textit{Algorithm \ref{alg:grid}}, for $ k=3 $. The blue block contains $ 3 $ Grover operators, one for each of finding marked elements in  $ X_1,X_2 $ and $ X_3 $. This block may be repeated several times (not shown in the figure) during the main iteration of Algorithm \ref{alg:grid}.
	} \label{fig:circuit}
\end{figure}
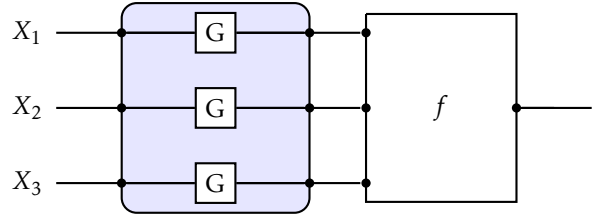

Now we briefly describe the working of the quantum grid search algorithm, Algorithm \ref{alg:grid}. At the heart of the algorithm is $ k $ Grover operator working in parallel. The local Grover operator $ G_i $ is set up to search for a marked item in $ X_i $ and it is implemented in the standard way using the local oracle function $ f_i $ defined in \meqref{Eq:localoracle}. The algorithm proceeds by initializing the values of parameter $ \lambda $ and the variable $ m $. The variable $ m $ is initially set to $ 1$,  and later in the algorithm it is adaptively scaled by $ \lambda $, by setting its value to $ \lambda m $, if the search for a state that satisfying the criterion set up by the function $ f $ is unsuccessful. The variable $ m $ effectively controls the number of  Grover's rotation that is performed, as  $ j $ gets selected randomly from the set $\{0,1,\cdots \ceil {m-1} \}  $, for $ i=1,2,\cdots k $. A priori, the number of marked elements in $ X_i $ is not known therefore $ m $ is adaptively scaled to ensure that the correct number of iterations of  Grover operator  is selected as required by Grover's search algorithm. 

\begin{algorithm}
	\caption{Quantum grid search with unknown number of solutions}
	{	Description: Let $ X_i = \{x_{ij} \}$,  with $ j $ running from $ 1 $ to $ n_i $,  be a set of cardinality $ n_i $. We also assume that out of $ n_i $ elements, a fixed but unknown number of $ m_i $ elements are marked. A ``path'', or equivalently an element $ (y_1,y_2,\cdots y_k) \in  X= \prod_{i=1}^{k} X_i  $, is understood to be marked, if, for $ i=1,2,\cdots k $, $ y_i $ is a marked element in $ X_i $. Suppose there exists a function $ f  $ from $ X= \prod_{i=1}^{k} X_i $ to $ \{0,1\} $ such that $ f(y_1,y_2,\cdots y_k) = 1 $ if and only if  $ y_i$ is a marked elements in $X_i$ for $ i=1,2,\cdots k $; otherwise $ f(y_1,y_2,\cdots y_k) =0 $.}
	\label{alg:grid}
	\begin{algorithmic}[1]
		\REQUIRE The set $ X$ and an oracle function $f$ with $ f(y_1,y_2,\cdots y_k) = 1 $ if and only if  $ y_i$ is a marked element in $X_i$ for $ i=1,2,\cdots k $; otherwise $ f(y_1,y_2,\cdots y_k) =0 $. 
		
		\ENSURE $x\in X$, a marked element. 
		
		\STATE Initialize $m \gets 1$ and  $ \lambda \gets 1 + \frac{1}{2}\left(\frac{4^k}{(4^k-1)} -1\right) $. (Any $ \lambda $, such that $ 1 < \lambda < \frac{4^k}{4^k -1 }  $, will work.)
		
		\FOR {$i = 1$ to $ k $}
		
		\IF{$m \leq \sqrt n_i$}
		
		\STATE pick uniformly random $j \gets \{0,\dots, \ceil{m-1}\}$.
		
		\STATE apply $j$ times the basic Grover iteration $G_i$ on initial state
		$\ket{\psi_i} = \sum_{x\in X_i}\frac{1}{\sqrt n_i} \ket{x}$.
		
		\ENDIF
		
		\ENDFOR
		
		\STATE Measure and obtain $f(y_1,y_2,\cdots y_k)$. If the outcome is $ 1 $, then output ``Success, $(y_1,y_2,\cdots y_k)$'' and
		exit. Otherwise, set $m \gets \lambda m$ and go to Step $ 2 $.

	\end{algorithmic}
\end{algorithm}

\begin{lem}
	Let $ m $ be a positive integer such that $$ m > \max\left(\frac{1}{\sin (2 \theta_1)},\frac{1}{\sin (2 \theta_1)},\cdots \frac{1}{\sin (2 \theta_k)} \right), \text{ where } \theta_i = \sin^{-1} \left(\sqrt{\frac{m_i}{n_i}} \right). $$
	Let $ \ds P_m = \text{the average probability of success of Algorithm \ref{alg:grid}} $, during the iteration wherein Grover operator $ G_i $ is applied $ j $ times, and $ j $ is chosen randomly from the set $\{0,1,\cdots \ceil {m-1} \}  $, for $ i=1,2,\cdots k $. 
	Then  $$ \ds  P_m  \geq \frac{1}{4^k}. $$
\end{lem}
\begin{proof}	
	From \meqref{eq_Grover_iteration}, the probability of finding a marked element in $ X_i$ after $j$ iterations is given by 
	\begin{equation*}
	p_{i,j}:=  \sin^2((2j+1)\theta_i) \, .
	\end{equation*}
	We note that the average probability of success, when the Grover operator $ G_i $ is applied $ j $ times and $ j $ is chosen randomly from the set $\{0,1,\cdots \ceil {m-1} \}  $, is given by
	\begin{align*}
	P_m & = \dfrac{1}{m^k} \sum_{j_1,j_2,\cdots j_k=0}^{m-1}  \prod_{i=1}^{k} p_{i,j_i} \\
	& = \dfrac{1}{m^k} \sum_{j_1,j_2,\cdots j_k=0}^{m-1}  \prod_{i=1}^{k} \sin^2((2j_i+1)\theta_i).\\
	& = \prod_{i=1}^{k} \dfrac{1}{m} \sum_{j_i=0}^{m-1} \sin^2((2j_i+1)\theta_i) \\
	& = \prod_{i=1}^{k} \dfrac{1}{2m} \sum_{j_i=0}^{m-1} 1 - \cos((2j_i+1)2\theta_i) \\
	& = \prod_{i=1}^{k} \left(\dfrac{1}{2} - \frac{\sin(4m \theta_i)}{4 m \sin (2 \theta_i) } \right).
	\end{align*}
	In the last step we have used the trigonometric identity
	\[
	\sum_{j=0}^{m-1} \, 1 - \cos \left( (2j+1 ) \theta\right) = m - \frac{1}{2} \frac{\sin 2 m \theta}{\sin \theta}.
	\]
	Since, $$ m > \max\left(\frac{1}{\sin (2 \theta_1)},\frac{1}{\sin (2 \theta_1)},\cdots \frac{1}{\sin (2 \theta_k)} \right),$$ we have  
	$$ \frac{\sin(4m \theta_i)}{4 m \sin (2 \theta_i) } \leq \frac{1}{4 m \sin (2 \theta_i) } \leq \frac{1}{4}.$$  Therefore, $$ \dfrac{1}{2} - \frac{\sin(4m \theta_i)}{4 m \sin (2 \theta_i) }  \geq \frac{1}{4}, \qquad \text{for $ i=1,2\cdots k $.}$$ 
	Hence, the average probability of success $\ds P_m \geq \frac{1}{4^k} $.
\end{proof}

\begin{thm} \label{thm:main}
	Algorithm $\ref{alg:grid}$ finds success within the expected running time of the order of $ \OO\left(\max\left(\sqrt{\frac{n_i}{m_i}}\right)_{i=1}^{i=k} \right) $.
\end{thm}
\begin{proof}
	First, we assume $ m_i \leq \, 0.75 n_i $, for $ i=1,2,\cdots k $.
	Let $$ \alpha_{i} = \frac{1}{\sin (2 \theta_i)} = \frac{n_i}{2 \sqrt{(n_i-m_i) m_i}} \le \sqrt{\frac{n_i}{m_i}}. $$ 
	The algorithm is said to be in the critical stage, if it
	goes through the main loop more than $ \ceil {\log_{\lambda} \alpha^{*}} $ times, where $ \alpha^{*} = \max(\alpha_1,\alpha_2,\cdots \alpha_k) $.
	
	During the $ i^{\text{th}} $ iteration of the main loop in Algorithm $\ref{alg:grid}$, the value of $ m $ is $ \lambda^{i-1} $. As Grover operator $ G_i $ is applied $ j $ times, and $ j $ is chosen randomly from the set $\{0,1,\cdots \ceil {m-1} \}  $, on the average the Grover operator $ G_i $ will run less than $ \frac{\lambda^{i-1}}{2} $ times during $ i^{\text{th}} $ iteration of the main loop. Also, there are $ k $ Grover operators are running in parallel. Therefore, an upper bound on the expected total number of Grover iterations needed to reach the critical stage, if the algorithm ever reaches it, is given by
	\begin{align}
	\frac{k}{2} \sum_{i=1}^{\ceil {\log_{\lambda} \alpha^{*}} } \lambda^{i-1} < \frac{k}{2} \frac{\lambda}{\lambda -1 } \alpha^{*} .
	\end{align}
	
	Therefore, the running time of the algorithm will be of $ \OO(\alpha^{*}) $, if it succeeds before reaching the critical stage.	
	On the other hand, if the critical stage is reached, every iteration thereafter will have a probability of success bigger than $ \frac{1}{4^k} $. Therefore,
	an upper bound on the expected number of Grover iterations needed for the algorithm to succeed, if the critical stage is reached, is given by
	\begin{align}
	\frac{k}{2} \sum_{s=0}^{\infty}  \left(1-\frac{1}{4^k}\right)^s \left(\frac{1}{4^k}\right) \lambda^{s + \ceil {\log_{\lambda} \alpha^{*}} } & < \frac{k \, \alpha^{*}}{2^{2k+1}} \sum_{s=0}^{\infty}  \left(1-\frac{1}{4^k}\right)^s  \lambda^{s +1  }  \nonumber \\
	& <  \left(\frac{k \lambda}{2^{2k+1}(1- (1-2^{-2k}) \lambda) }\right) \alpha^{*}.
	\end{align}
	It is clear that the run time of the algorithm is of the order of $ \OO(\alpha^{*}) $, as desired. The remaining case, when $ m_i > 0.75 \, n_i $, can be easily dealt with using a classical sampling in a constant expected time.
\end{proof}

\section{Trajectory optimization: Problem formulation} \label{sect:problem formulation}
We describe below the problem formulation given in \cite{shuklavedula2019}.
Suppose $ \vec{U}(t) \in \RR^m  $ and $ \vec{X}(t) \in \RR^n  $ denote the control function  and the corresponding state trajectory respectively at time $ t $. The goal is to determine the optimal control function $ \vec{U}(t)$ and the corresponding state trajectory $ \vec{X}(t)$ for $ \tau_0 \leq t \leq \tau_f $ such that the following Bolza cost function is minimized:
\begin{align}\label{Eq_Main_Cost}
\mathcal{J}(\vec{U}(.),\vec{X}(.),\tau_f) = \mathcal{M}(\vec{X}(\tau_f),\tau_f) + \int_{\tau_0}^{\tau_f} \, \mathcal{L}(\vec{U}(\tau),\vec{X}(t),t) \, dt.
\end{align}
Here $ \mathcal{M} $ and $ \mathcal{L} $ are $ \RR $-valued functions. Moreover, we also assume that the system satisfies the following dynamic constraints. 
\begin{align}
f_l \leq f(\vec{U}(\tau),\vec{X}(t),\vec{X^'}(t),t) \leq f_u \qquad t \in [\tau_0, \tau_f].
\end{align}
In addition, we also specify the following boundary conditions 
\begin{align}
h_l  \leq h(\vec{X}(\tau_0),\vec{X}(\tau_f),\tau_f -\tau_0) \leq h_u,
\end{align}
where $ h $ an $ \RR^p $-valued function and $ h_l, h_u \in \RR^p$ are constant vectors providing the lower and upper bounds of $ h $. Finally, we note the mixed constraints on control and state variables
\begin{align}
g_l \leq g(\vec{U}(t),\vec{X}(t),t) \leq g_u ,
\end{align}
with $ g $ a $ \RR^r $-valued function and $ g_l, g_u \in \RR^r$ are constant vectors providing the lower and upper bounds of $ g $.

\subsection{Discretization approaches} \label{Discretization}
We refer readers to \cite{shuklavedula2019} for our discretization approaches.
We note that with this discretization the trajectory optimization problem is transformed into a discrete search problem 
and where the objective function to be minimized is the discrete form of  \meqref{Eq_Main_Cost}.

In principle, one can now solve the problem by traversing the discrete search space and finding the state that corresponds to the minimum cost and satisfies all the imposed constraints. Still, in practice, the problem is the massive size of the discrete search space. The quantum grid search algorithm offers a promising method to tackle such problems. 

\section{Application of quantum grid search algorithm in trajectory optimization} \label{sect:application}
The quantum grid search algorithm (Algorithm~\ref{alg:grid}) described earlier, is applicable in the general context of the problem of finding the minimum cost for the problem described in \mref{Discretization}, with the cost function being the discrete version of \meqref{Eq_Main_Cost}. We explain in the following how the quantum grid search algorithm can be used to solve a trajectory optimization problem under some mild assumptions. First, we will consider a two-dimensional problem and then subsequent we will describe how a general $ d$-dimensional problem, with $ d >2  $ can be tackled using our  algorithm. For illustrating the key idea we will first consider the brachistochrone problem, with the objective of minimizing the time $ \tau_f $ given by a discrete version of \meqref{eq_cost}, and then give the most general case.

Before proceeding further, we  define precisely some relevant terms. The definitions given below are applicable for any two-dimensional trajectory optimization problem. However, it makes more intuitive sense in the context of the brachistochrone problem described in \mref{subsect:brachistochrone}.
\begin{defn}[Path] \label{def:path}
	An element $ (y_1,y_2,\cdots y_k) \in \prod_{i=1}^{k} X_i $ is called a path.
\end{defn}

\begin{defn}[Cost]
	$ \Cost $ is a real valued function $ \{\Cost \ | \ \prod_{i=1}^{k} X_i \to \RR \} $, such that $ \Cost\left((y_1,y_2,\cdots y_k)\right) $ gives the cost  associated to the path $ (y_1,y_2,\cdots y_k) $ (or more precisely, for the continuous path $ y(x) $ obtained via interpolation from the path $ (y_1,y_2,\cdots y_k) )$. For example, for the brachistochrone problem the cost is calculated using \meqref{eq_cost}.
\end{defn}
We remark here that although we are using the `term' cost, in general it could be any quantity of interest, such as in the context of the brachistochrone problem it is basically the `time'.

\begin{defn}[$\SolutionPaths  (a,b) $]
	A path $ (y_1,y_2,\cdots y_k) \in \prod_{i=1}^{k} X_i $ belongs to the set $\SolutionPaths  (a,b) $, if $ a < \Cost\left((y_1,y_2,\cdots y_k)\right) < b $. In other words,
	\begin{align}
	\SolutionPaths (a,b) = \left\{(y_1,y_2,\cdots y_k) \in \prod_{i=1}^{k} X_i \,  \ | \ \,  a < \Cost\left((y_1,y_2,\cdots y_k)\right) <  b  \right\}.
	\end{align}
\end{defn}
\begin{defn}[Marked Element]
	An element 	$ x_{i,j} \in X_i $ is considered a marked element of $ X_i $ with respect to $\SolutionPaths (a,b) $, if there exist a path say $ (y_1,y_2,\cdots y_k) \in \SolutionPaths (a,b)$, such that $ y_i = x_{i,j} $.
\end{defn}


\subsection{\textbf{Brachistochrone problem}} \label{subsect:brachistochrone}
The objective in the brachistochrone problem is to find the required trajectory of a particle starting from the rest under the influence of gravity, without any friction, such that it slides from the one fixed point to the other in the least possible time.
Let $ \vec{X}(t) = x(t) \vec{i} + y(t) \vec{j} $ be the position of the particle. The goal is to minimize $ \tau_f $ given by
\begin{align} \label{eq_cost}
\tau_f = \int_{0}^{\pi} \, \sqrt{\frac{1 +(\frac{dy}{dx})^2 }{2gy}} \, dx,
\end{align}
under the boundary conditions  
$ (x(\tau_0),y(\tau_0)) = (0,2) $ and $ (x(\tau_f),y(\tau_f)) = (\pi,0) $ with $ \tau_0 = 0 $.
We describe a possible discretization approach to solve the brachistochrone problem.

\subsubsection{\textbf{Global discretization in physical space}} \label{sub_sect_discrete_physical_space}

Out of infinitely many possible paths, the goal is to pick the path with the minimum time. For practical purposes it is sufficient to find a `good enough' approximate solution (based on acceptable levels of errors).
The discretization is chosen based on what constitute  a `good enough' solution in a given context. In our present example, the physical space consists of the rectangle $ [0,\pi] \times [0,2] $ in $ \RR^2 $. Of course, one may as well consider a bigger rectangle to allow for the possibility of better solutions. One can now discretize the rectangle $ [0,\pi] \times [0,2] $ in several possible ways. For example, let $ x_i = \frac{ \pi i}{k} $ for $i = 1 \cdots k $. Let $ y_{i} \in \{ x_{i,j} = \frac{2j}{n_i} :\: j = 0,1 \cdots n_i \} $ for $ i=1 $ to $ n_i $. We set $ y_1 =2 $  and $ y_k = 0 $ to account for the boundary condition. 

Essentially, it means that once a discretization of $ x $ is carried out then we require that at $x= x_i $, the corresponding $ y(x_i) $ can only have values from the set $ \{ x_{i,j} :\: j = 0,1 \cdots n_i \} $ (see \mfig{Fig_grid}). 
We remark here that the definition of ``path'' given in Def.~\ref{def:path} makes sense, as given $ (y_1,y_2,\cdots y_k) \in \prod_{i=1}^{k} X_i $, one can use Lagrange interpolation to find a continuous path $ y(x) $ such that $ y(x_i) = y_i $, for $ i=1,2,\cdots k $, and $ y(x) $ is then a candidate solution of the brachistochrone problem under consideration. We note that with the above discretization there are $ N = \prod_{i=1}^{k} n_i $  total possible cases need to be considered. The minimum time can now be obtained by considering all the paths in this search space.

\begin{figure}
	\centering	
	\begin{tikzpicture} [scale=0.75]
	\draw [step=1cm,blue!40,thin] (0,0) grid (9,10); 
	\foreach \k in {1,2,3}
	\node  at (\k-1,-0.5) {$ x_{\k} $};
	\node  at (3,-0.5)  {$ \cdots  $};
	\node  at (4,-0.5)  {$ \cdots  $};
	\node  at (5,-0.5)  {$ x_i  $};
	\node  at (6,-0.5)  {$ x_{i+1}  $};
	\node  at (7,-0.5)  {$ \cdots  $};
	\node  at (8,-0.5)  {$ \cdots  $};
	\node  at (9,-0.5)  {$ x_k  $};
	\foreach \k in {2,3,4,5}{
		\pgfmathtruncatemacro \y {2*\k};
		\node  at (5.5,\k-0.75) {{\small $ x_{i,\k}$}};}
	\node [rotate=90] at (5.5,5)  {$ \cdots  $};
	\node  at (5.5,6.25)  {{\small $x_{i,i}    $}};
	\node  at (5.5,7.25)  {{\small $ x_{i,i+1}  $}};
	\node [rotate=90] at (5.5,8)  {$ \cdots  $};
	\node at (5.5,9.25) {{\small $ x_{i,n_i-1}  $}};
	\node at (5.5,10.25) {{\small $ x_{i,n_i}$}};
	\node at (5.5,0.25) {{\small $ x_{i,1} $}};
	\foreach \i in {0,...,10}{
		\draw [color=black, very thick] (5,\i) circle (1.5pt) ;
	}
	\end{tikzpicture}
	\caption{Brachistochrone problem: discretization of both the dependent variable $ y $ and the independent variable $ x $ is carried out in our approach. Note that at $ x=x_i $, the corresponding $ y(x_i) $ is allowed to take values only from the set  $ \{ x_{i,j} :\: j = 0,1 \cdots n_i \} $. The possible values for $ y(x_i) $ are shown with thick black dots.} \label{Fig_grid}
\end{figure}
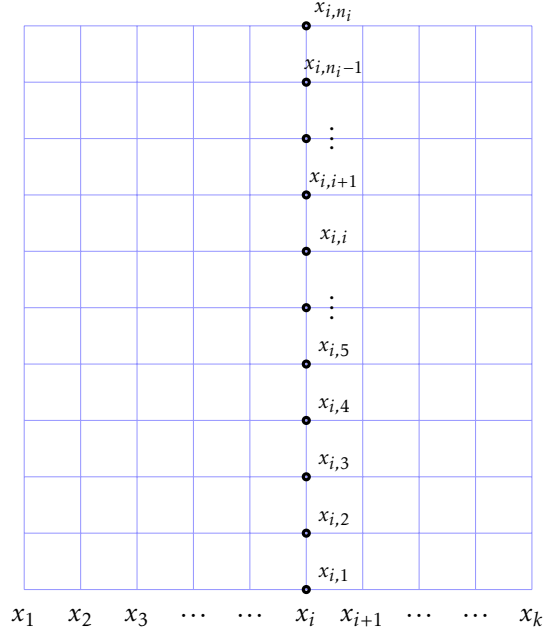

If needed, one can always approach the correct solution to the desired accuracy by making the discrete grid finer. However, the discrete search space grows very rapidly, therefore,  an efficient algorithm, such as the quantum grid search algorithm described in Sect.~\ref{sect:grid_search}, is needed to solve the problem.

\subsection{Solution of Brachistochrone problem}
In the following we give Algorithm~$\ref{alg:binary} $, which uses  Algorithm~$ \ref{alg:grid}$, and the idea of a classical binary search, to successively improve the bounds $ (a,b) $ such that $\SolutionPaths  (a,b) $ is non-empty.

We assume that an upper bound $ b_0 $ for the minimum time is known. For example, this can be achieved by carrying out a classical measurement after selecting a random path $ (y_1,y_2,\cdots y_k) \in \prod_{i=1}^k X_i$ and calculating the associated cost $ \Cost\left((y_1,y_2,\cdots y_k)\right) $. Then $ b_0 = \Cost\left((y_1,y_2,\cdots y_k)\right) $. We also require a choice for the desired possible lower bound on the minimum time $ a_0 $. Of course, $ a_0$ could just be set to $ 0 $.  Next, we define oracle function $ f_{a,b} $ as follows
\begin{align*}
f_{a,b}(y_1,y_2,\cdots y_k) = \begin{cases}
& 1  \qquad \text{if, $ a < \Cost\left((y_1,y_2,\cdots y_k)\right) < b  $,}\\
&  0 \qquad \text{otherwise.}
\end{cases}
\end{align*}

\begin{algorithm}
		\caption{A binary search algorithm based on  Algorithm~$ \ref{alg:grid}$ }
			\label{alg:binary}
	\begin{algorithmic}[1]
			
		\STATE Initialize $a=a_0$, $ b=b_0 $, count $=0 $, c $ = $ max-count .
		
		
		\WHILE{count $ < $ c }
		
		\STATE Set $ m = \frac{a+b}{2} $ and count $=$ count $ + 1$.

		\IF {Algorithm~$ \ref{alg:grid} $ called with the oracle function $ f_{a,m} $ is successful}
		
		\STATE \qquad $ b = m $
		
		\ELSIF {Algorithm~$ \ref{alg:grid} $ called with the oracle function $ f_{m,b} $ is successful}
		\STATE \qquad $ a = m $
		
		\ELSE 
		
		\STATE  Exit with output $ (a,b) $ 
		
		\ENDIF
		
		\ENDWHILE
		
	    \STATE	Exit with output $ (a,b) $ 
		
	\end{algorithmic} 
\end{algorithm}

\vspace{1 cm}
\paragraph{Advantages:} We now discuss the advantages offered by Algorithm~$\ref{alg:grid} $ and Algorithm~$\ref{alg:binary} $. First of all, we note that the set $ X = \prod_{i=1}^{k} n_i $ has the cardinality of $ \prod_{i=1}^{k} n_i $, i.e., there are possible $\prod_{i=1}^{k} n_i  $ number of paths. For ease of analysis assume that $ n_i = N $ for all $ i$. Then there are total $ N^k $ paths, with each path having an associated cost (or time in this case). The goal here is to find the path with the minimum associated cost, out of these $ N^k $ paths. A classical algorithm will clearly need $ \OO(N^K) $ searches in the worst case. If we use any known minimization algorithms  based on Grover's search algorithm directly, treating each path as an input, then with a quadratic speed-up the quantum algorithm will succeed with probability close to $ 1 $ within $ \OO(\sqrt{N^k}) $. Now we consider our quantum grid search algorithm, i.e., Algorithm~$\ref{alg:grid} $. Clearly in the situation under consideration, in the worst case the running time of Algorithm~$\ref{alg:grid} $ is of the order of  $ \OO(\sqrt{N})  $. Finally, Algorithm~$\ref{alg:binary} $ calls Algorithm~$\ref{alg:grid} $, c number of times. Therefore, Algorithm~$\ref{alg:binary} $ will succeed in $ c \, \OO(\sqrt{N}) $. Of course, Algorithm~$\ref{alg:binary} $ does not actually find the minimum solution, it just gives the bound $ (a,b) $ on the solution, and given $ (a,b) $ a direct application of Algorithm~$\ref{alg:grid} $ with the Oracle function $ f_{a,b} $, a path in $\SolutionPaths  (a,b) $ can be discovered easily. However, it should be clear that in many situations Algorithm~$\ref{alg:grid} $ and Algorithm~$\ref{alg:binary} $ offer great computational advantages over the classical methods. In fact, for some practical applications, if one is not looking for an absolute minimum, but just a `good enough' solution path with the associated  cost less than a predetermined value, then Algorithm~$\ref{alg:binary} $ will be even more efficient in discovering such a solution.

\subsection{Generalization to higher-dimensional problems}
Now we consider how  Algorithm~$\ref{alg:grid} $ and Algorithm~$\ref{alg:binary} $ could be employed to treat higher dimensional optimization problems.
Suppose the discretization is carried out as described in \cite{shuklavedula2019}. This will result in a discrete search space (similar to the discrete grid described earlier for the brachistochrone problem, but possibly of a higher dimension with a large value of $ k $ and with the role of $ x $ variable being played by the time variable $ t $). To make this point clear we revisit the two dimensional trajectory, say $ y=f(x) $, used in the brachistochrone problem. One of the key features in our discretization scheme is that we discretize both the independent variable $ x $ as well as the dependent variable $ y $. Which means we take samples at $ x=x_i $ for $ i=1,2,\cdots k $, and for each $ x_i $ then the corresponding $ y $-coordinate, i.e.,  $ y_i = f(x_i) $, can only take values from a finite set. We refer to \mfig{Fig_grid}  which illustrates this and wherein $ y_i = f(x_i) $ can only take values from the set $ X_i = \{ x_{i,j} :\: j = 0,1 \cdots n_i \} $. Clearly, each choice of an element in $ (y_1,y_2,\cdots y_k) \in \prod_{i=1}^{k} {X_i} $ defines a ``path'' (one can use Lagrange interpolation using the points $ (x_i,y_i) $ for $ i=1,2,\cdots k $ to construct the continuous path $ y=f(x) $ passing through these points). Now consider motion in the $ 3 $-dimensional space. Of course, mathematically a curve with the parametric equation  $ r(t) = x(t) \vec{i} + y(t) \vec{j} + z(t) \vec{k} = \langle x(t),y(t),z(t) \rangle $ for $ t \in [t_0,t_f] $ can describe this motion. Now we discretize the parameter $ t $ by taking samples at $ t = t_i $ for $ i=1,2,\cdots k $, and for each $ t_i $ then the corresponding $ x,y$ and $z$-coordinates, i.e.,  $ x_i = x(t_i),y_i=y(t_i) $ and $ z=z(t_i) $ can only take values from finite sets, $ X_i,Y_i$ and $ Z_i $, respectively. Clearly this means the motion of the particle is restricted on a  $ 3 $-dimensional discrete grid. Now the concept of ``path'' and the ``marked path'' can easily be defined. Therefore, essentially the problem formulation in this case will be quite similar to the formulation in the brachistochrone problem described  earlier. Indeed, it is clear that in order to implement our quantum grid search algorithm , i.e., Algorithm~$\ref{alg:grid} $, in this case the set $ X $ therein should be replaced by $ \prod_{i=1}^{k} X_i Y_i Z_i $, with each element in  $ \prod_{i=1}^{k} X_i Y_i Z_i $  defining a ``path''. Then for each of $ X_i ,Y_i$ and $ Z_i $ for $ i=1,2,\cdots k $ a Grover operator will be required for implementing Algorithm~$\ref{alg:grid} $. 

From the previous discussion, it should be now clear that this method could be generalized, similarly, to even higher dimensional trajectory optimization problems.  Indeed, for each discretized state and control variable at chosen sampling points, a Grover operator will be needed, and all such Grover operator will work in parallel. For example, suppose $ X_i $, for $ i=1 $ to $ p $, are $ p $ state variables, and $ Y_i $, for $ i= 1$ to $ q $, are $ q $ control variables. In this case, in Algorithm~$\ref{alg:grid} $, the set $ X $ should be replaced by $ \prod_{i=1}^{k} \prod_{r=1}^p  \prod_{s=1}^{q} {X_r}_i {Y_s}_i $, with each element in $ \prod_{i=1}^{k} \prod_{r=1}^p  \prod_{s=1}^{q} {X_r}_i {Y_s}_i $  defining a ``path''. Therefore, we conclude that Algorithm~$\ref{alg:grid} $ is applicable to higher dimensional trajectory optimization problems.

Here we want to point out a limitation to our proposed approach. It is obvious that for a high dimensional problem, with a very large value of $ k $, the number of Grover operators needed will be very big. This means the number of components and logic gates needed for the quantum circuit implementing  Algorithm~$\ref{alg:grid} $ will be very high. In a sense the gain in computational speed is limited by the increased complexity of the quantum circuit. 
Still, even with this limitation many optimizations problems will be amenable to being treated by  Algorithm~$\ref{alg:grid} $ and Algorithm~$\ref{alg:binary} $.

\section{\textbf{Conclusion}} \label{sect:Conclusion}

In this paper, we proposed two quantum search algorithms, Algorithm~$\ref{alg:grid} $ and Algorithm~$\ref{alg:binary} $. The first algorithm, Algorithm~$\ref{alg:grid} $, addresses a class of special search problems involving identification of ``marked elements'' associated with specified bounds on a discrete grid. This class of problems is relevant to many fields, including global trajectory optimization. The second algorithm, Algorithm~$\ref{alg:binary} $, is a binary search algorithm based on  Algorithm~$ \ref{alg:grid}$), that efficiently solves global trajectory optimization problems. We note that the key idea behind Algorithm~$\ref{alg:grid} $ is the use of parallel Grover operators for quantum search. We proved (using Theorem~$ \ref{thm:main} $,) that Algorithm $\ref{alg:grid}$ finds success within the expected running time of  $\OO\left(\max\left(\sqrt{\frac{n_i}{m_i}}\right)_{i=1}^{i=k} \right) $. This is a big improvement over a pure non-parallel classical algorithm which is of the order of  $ \prod_{i=1}^{k} n_i $. It is important to note that such enhanced computational efficiency of our proposed algorithms comes at a cost of increased complexity of the underlying quantum circuitry. We also note that our assumption on the availability of local oracle functions may not be true in many cases. Still, if these conditions are satisfied, then the computational advantages offered by our quantum algorithms could potentially enable us to tackle very high dimensional problems, which are beyond the capability of classical computers. The proposed algorithms are especially valuable in cases where solutions within specified bounds might be sufficient for practical use, instead of absolute global optimum solutions that can be difficult to find.  
Although, quantum computing is still in its infancy, and limited by the number of qubits available, hopefully, in the future with the advent of more powerful quantum computers the practical implementation of the algorithms presented in this work will become feasible.

\newpage

		\bibliographystyle{plain}
		
		\bibliography{Bibliography}

\end{document}